\documentclass[letterpaper, 10 pt, conference]{ieeeconf}  

\IEEEoverridecommandlockouts                              
                                 
\makeatletter

\let\proof\@undefined
\let\endproof\@undefined
\let\IEEElabelindent\labelindent
\makeatother

\usepackage{graphicx}
\usepackage{caption}
\usepackage{subcaption}
\usepackage{enumerate}
\usepackage{hyperref}
\usepackage{upgreek}
\usepackage{array}
\usepackage{multirow}

\usepackage{amsmath} 
\usepackage{amssymb}  
\usepackage{amsthm}  
\usepackage{bm}
\usepackage{lipsum}
\usepackage[linesnumbered, ruled]{algorithm2e}
\usepackage{color}
\usepackage{cite}

\newtheorem{proposition}{Proposition}

\newtheorem{assumption}{Assumption}
\newtheorem{lemma}{Lemma}
\newtheorem{theorem}{Theorem}

\newcolumntype{M}[1]{>{\centering\arraybackslash}m{#1}}

\title{\LARGE \bf
Wasserstein Distributionally Robust Control of Partially Observable Linear Systems: Tractable Approximation and Performance Guarantee
\thanks{This work was supported in part by  the National Research Foundation of Korea funded by  MSIT(2020R1C1C1009766), the Information and Communications Technology Planning and Evaluation (IITP) grant funded by MSIT(2020-0-00857), and Samsung Electronics.}
}

\author{Astghik Hakobyan \and Insoon Yang
\thanks{A. Hakobyan, and I. Yang are with the Department of Electrical and Computer Engineering and ASRI, Seoul National University, Seoul, 08826, South Korea {\tt\small \{astghikhakobyan, insoonyang\}@snu.ac.kr}}%
}

\begin{document}

\maketitle
\thispagestyle{empty}
\pagestyle{empty}


\begin{abstract}
Wasserstein distributionally robust control (WDRC)  is an effective method for addressing inaccurate distribution information about disturbances in stochastic systems. It provides various salient features, such as an out-of-sample performance guarantee, while most of the existing methods use full-state observations. 
In this paper, we develop a computationally tractable WDRC method for discrete-time partially observable linear-quadratic (LQ) control problems.
The key idea is to reformulate the WDRC problem as a novel minimax control problem with an approximate Wasserstein penalty. 
We derive a closed-form expression of the optimal control policy of the approximate problem using  a nontrivial Riccati equation. 
We further show the guaranteed cost property of the resulting controller and identify a provable bound for the optimality gap. 
Finally, we evaluate the performance of our method through numerical experiments using both Gaussian and non-Gaussian disturbances.
\end{abstract}


\section{Introduction}\label{sec:intro}

The problem of finding an optimal control policy for a partially observable dynamical system under uncertainties, such as  disturbances, is an important and well-studied topic, for which numerous results have been obtained.
Classically, the optimal control of partially observable systems under uncertainties is regarded either in stochastic or in robust control frameworks.  
Stochastic optimal control methods seek the control policies that minimize an expected cost of interest under the disturbance distribution known \emph{a priori}~\cite{aastrom2012introduction}. 
The most well-known method of this type with partial state observations is the linear-quadratic-Gaussian (LQG) control. 
For linear systems, it is common practice to use the separation principle to independently design $(i)$ a feedback law as if the full state observation is available
and $(ii)$ a state estimator~\cite{witsenhausen1971separation}.
However, LQG control assumes that the disturbance follows 
a Gaussian noise model with a known mean and covariance matrix and thus ignores possible inaccuracies in the distribution information. 
By contrast, robust control methods aim to design a controller that minimizes the worst-case cost assuming the disturbances lie in known ranges~\cite{Doyle1996}. 
 Thus, these methods disregard the potentially useful statistical properties of disturbances.

Distributionally robust control (DRC) is
an alternative method that bridges the gap between the two aforementioned  approaches~(e.g., \cite{petersen2000minimax, Ugrinovskii2002, 
van2015distributionally,
Tzortzis2015, 
 coulson2019regularized,  Schuurmans2020, Mark2020, yang2020wasserstein, Tzortzis2021, zolanvari2022data, Dixit2022}). 
The aim is to obtain a controller that minimizes the expected value of a given cost function with respect to the worst-case distribution chosen from an \emph{ambiguity set}. 
In the literature on distributionally robust optimization (DRO), it is common to construct the ambiguity set based on a nominal distribution estimated from observations so that it contains the true distribution with high probability~(see~\cite{Rahimian2019} and the references therein).
Among various types, 
the ambiguity set using Wasserstein distance has recently received great attention as a tool for hedging against data perturbations and distributional uncertainties. 
The Wasserstein ambiguity set  is a statistical ball centered at a nominal distribution with a radius measured by the Wasserstein metric. 
The key benefit is its ability to avoid pathological solutions to DRO problems, unlike the $\phi$-divergences~\cite{gao2022distributionally}. 
 Moreover, recent works on DRO using Wasserstein ambiguity sets demonstrate  computational tractability, as well as an out-of-sample performance  guarantee of the resulting solution~\cite{mohajerin2018data, gao2022distributionally, kuhn2019wasserstein}.
 Inspired by the success of the Wasserstein DRO, 
a few DRC methods using the Wasserstein ambiguity sets have been proposed~\cite{yang2020wasserstein, coulson2021distributionally, Mark2020, Kim2021, Hakobyan2021}.
However, existing works in Wasserstein DRC (WDRC) mostly assume full-state observations, ignoring the issue of partial observations.\footnote{A notable exception is the work of Coulson {\it et al.}~\cite{coulson2021distributionally} that proposes an MPC method using input and output data.}

In this work, we propose a novel WDRC method 
for discrete-time partially observable linear systems.
The key idea is to approximate the  WDRC problem as a  minimax  control problem with an approximate Wasserstein penalty that replaces the Wasserstein ambiguity set. 
Under an assumption on the penalty parameter, we recursively show that the value functions of the approximate problem have a quadratic form, whose coefficients are obtained via a nontrivial Riccati equation and a tractable semidefinite programming (SDP) problem. 
Moreover, we show that the approximate problem admits a unique optimal control policy, and we derive its closed-form expression (Section~\ref{sec:app}). 
Remarkably, the resulting controller is shown to enjoy a guaranteed cost property under any disturbance distribution chosen from the Wasserstein ambiguity set. 
This demonstrates the distributional robustness of our control method. 
We also identify a provable bound for the performance gap between our controller and the optimal distributionally robust (DR) controller (Section~\ref{sec:perf}). 
Finally, the results of our experiments demonstrate the capability of the proposed method to immunize the system against imperfect distribution information for both Gaussian and non-Gaussian disturbances (Section~\ref{sec:exp}). 


\section{Preliminaries}\label{sec:prel}

\subsection{Notation}

Let $\mathbb{S}^n$ denote the space of all symmetric matrices in $\mathbb{R}^{n\times n}$. We use $\mathbb{S}_+^n$ ($\mathbb{S}_{++}^n$) to represent the cone of symmetric positive semidefinite (positive definite) matrices in $\mathbb{S}^n$. 
For any $A,B\in\mathbb{S}^n$, the relation $A\succeq B$ ($A\succ B$) means that  $A - B\in\mathbb{S}_+^n$ ($A - B\in\mathbb{S}_{++}^n$). 
The set of Borel probability measures with support $\mathcal{W}$ is denoted by $\mathcal{P}(\mathcal{W})$.

\subsection{Problem Setup}
Consider a discrete-time linear stochastic system of the form
\[
\begin{split}
x_{t+1} &= Ax_t + B u_t + w_t\\
y_t & = C_t x_t + v_t,
\end{split}
\]
where $x_t\in\mathbb{R}^{n_x}$, $u_t\in\mathbb{R}^{n_u}$, and $y_t\in\mathbb{R}^{n_y}$ are the system state, input, and output at stage $t$, respectively. Here, $w_t\in\mathbb{R}^{n_x}$ and $v_t\in\mathbb{R}^{n_y}$ are random vectors, representing the system disturbance and the output noise, respectively. The initial state is drawn from some state distribution $f_x(x_0)$. Furthermore, the random vectors $w_t$, $v_t$ and $x_t$ are assumed to be independent. The observation noise $v_t$ is assumed to follow a zero-mean Gaussian distribution with a covariance matrix $M$.

In the partially observable setting, the information collected so far is given by
\[
\begin{split}
I_t &:= (y_0,\dots, y_t, u_0, \dots, u_{t-1}),\quad t=1,\dots,T,\\
I_0 &:= y_0,
\end{split}
\]
where $I_t$ is called the \emph{information vector}. 
After applying the control input to the system and observing the output, the information vector is updated according to
\[
I_{t+1} = (I_t, y_{t+1}, u_t),
\]
which can be viewed as a new dynamical system with state $I_t$.

In practice, it is often restrictive to assume that the probability distribution of $w_t$ is available. 
Let $\mathbb{P}_t$ and $\mathbb{Q}_t$ denote the \emph{unknown} true  and the nominal (or estimated) distribution measures of $w_t$, respectively. 
Our goal is to design a cost-minimizing controller that is robust against the deviation of $\mathbb{P}_t$ from $\mathbb{Q}_t$ in the partially observable case. 
For this, we take a game-theoretic approach.

Consider a two-player zero-sum game, in which Player I is the controller and Player II is a hypothetical adversary. 
Let $\pi_t$ denote the control policy at time $t$, mapping the information vector $I_t$ to a control input $u_t$. 
The policy for the adversary is denoted by $\gamma_t$ and  maps the information vector $I_t$ to a distribution measure $\mathbb{P}_t$ of $w_t$ from an \emph{ambiguity set} $\mathcal{D}_t \subset \mathcal{P}(\mathbb{R}^{n_x})$. 
The ambiguity set (to be defined later) is chosen to contain all the relevant probability measures to appropriately characterize distribution errors. 
To design a finite-horizon controller, the following cost functional is chosen:
\begin{equation}\label{new_cost}
\begin{split}
J(\pi, \gamma) &=  \mathbb{E}_{\bold{y}} \Big[\mathbb{E}_{x_T}[x_T^\top Q_f x_T \mid I_T] \:  \\ 
&
+ \sum_{t=0}^{T-1} \mathbb{E}_{x_t, w_t}[  x_t^\top Q x_t + u_t^\top R u_t \mid I_t, u_t]\Big],
\end{split}
\end{equation}
where $Q\in\mathbb{S}_+^{n_x}, Q_f\in\mathbb{S}_{+}^{n_x}, R\in\mathbb{S}_{++}^{n_x}$ and $T$ is the length of the time-horizon,  and outer expectation is with respect to the joint distribution of all measurements $\bold{y} := (y_0,\dots, y_{T})$.
Player I aims to find a control policy $\pi^{opt}$ to minimize the cost while Player II selects a distribution policy $\gamma^{opt}$ to maximize the same cost. 
The desired control policy can be obtained by solving the following DRC problem:
\begin{equation}\label{drc}
    \min_{\pi\in\Pi}\max_{\gamma\in\Gamma_\mathcal{D}} J(\pi,\gamma),
\end{equation}
where
$\Pi := \{\pi := (\pi_0, \ldots, \pi_{T-1}) \mid \pi_t (I_t) = u_t \}$ and
$\Gamma_{\mathcal{D}} := \{\gamma := (\gamma_0, \ldots, \gamma_{T-1}) \mid \gamma_t (I_t) = \mathbb{P}_t \in \mathcal{D}_t \}$
are the policy spaces for Players I and II, respectively. 
Note that the ambiguity set  is embedded in the policy space for Player II.

\subsection{Measuring Distribution Errors via the Wasserstein Metric}

In the DRO literature, it is popular to construct the ambiguity set  as a statistical ball containing distributions close to the nominal one~\cite{guo2019distributionally, bayraksan2015data, mohajerin2018data, gao2022distributionally, kuhn2019wasserstein}.
Therefore, characterizing the closeness of two probability distributions is the key component in designing the ambiguity set. In this work, we quantify the distance between two probability measures  via the Wasserstein metric. The Wasserstein distance of order $p$ between two measures $\mathbb{P}$ and $\mathbb{Q}$ supported on $\mathcal{W} \subseteq \mathbb{R}^n$ represents the minimum cost of moving the probability mass from one measure to another and is defined as
\[
\begin{split}
W_p(\mathbb{P}, \mathbb{Q}) := \inf_{\kappa \in \mathcal{P}(\mathcal{W}^2)} \bigg\{ \bigg(\int_{\mathcal{W}^2}&\|x-y\|^p \, \mathrm{d}\kappa(x,y) \bigg)^{1/p} \\
&\bigg| \, \Pi^1 \kappa = \mathbb{P}, \Pi^2 \kappa = \mathbb{Q} \bigg\},
\end{split}
\]
where  $\kappa$ is the \emph{transport plan}, with $\Pi^i\kappa$ denoting its $i$th marginal, and $\|\cdot\|$ is a norm on $\mathbb{R}^n$ that quantifies the transportation cost. 

Using the Wasserstein metric of order $p=2$ with $\|\cdot\|$ representing the standard Euclidean distance, we define the ambiguity set as the following statistical ball of radius $\theta>0$:
\[
\mathcal{D}_t: = \{\mathbb{P}_t \in \mathcal{P}(\mathbb{R}^{n_x}) \mid W_2(\mathbb{P}_t,\mathbb{Q}_t) \leq \theta\}.
\]
The ambiguity set contains all probability measures whose Wasserstein distance from the nominal one $\mathbb{Q}_t$ is no greater than $\theta$.\footnote{The radius $\theta$ is an important factor that determines the conservativeness of the resulting control policy. Calibrating $\theta$ is an important research topic, for which theoretical and empirical studies have been performed~(e.g., \cite{mohajerin2018data, Boskos2020, yang2020wasserstein}). }
As mentioned in Section~\ref{sec:intro}, Wasserstein ambiguity sets have superior statistical properties compared to other types and provide  high-performance controllers in various practical DRC problems~(e.g.,~\cite{Mark2021, coulson2021distributionally, hakobyan2021wasserstein, Hakobyan2021map}).

\section{Tractable Approximation and Solution}\label{sec:app}

The WDRC problem~\eqref{drc} is challenging to solve, particularly when the dimension of the state space is large. 
A dynamic programming (DP) approach can be used together with a tractable reformulation of the Bellman equation in a similar fashion to the full-state observation case~\cite{yang2020wasserstein}.
However, the computational complexity of the DP method increases exponentially with the dimension of the state space. 
To develop a scalable solution, we first propose an approximate version of the WDRC problem.  We then derive a Riccati equation for the approximate problem with the corresponding closed-form expression for the optimal policy.

\subsection{Approximation with Wasserstein Penalty}

Instead of using the Wasserstein ambiguity set, 
our approximation employs a Wasserstein penalty to penalize the deviation from the nominal distribution, motivated by our previous work for the full observation case~\cite{Kim2021}. 
Specifically, consider the following modified cost function
with an additional Wasserstein penalty term:
\begin{align}
&\mathbb{E}_{\bold{y}}\Big[\mathbb{E}_{x_T}[x_T^\top Q_f x_T \mid I_T] \: + \label{obj}\\ 
&\sum_{t=0}^{T-1} \Big(\mathbb{E}_{x_t, w_t}[ x_t^\top Q x_t + u_t^\top R u_t \mid I_t, u_t]- \lambda W_2(\mathbb{P}_t,\mathbb{Q}_t)^2\Big)\Big],\nonumber
\end{align}
where $\lambda >0$ is the penalty parameter used to adjust  the conservativeness of the control policy.

Unfortunately, the partially observable minimax control problem with the Wasserstein penalty is intractable unlike the full observation case. 
To resolve this issue, we use the following approximation of the Wasserstein penalty.

\begin{lemma}\label{lem:bound}
Let 
\[
\bar{w}_t := \mathbb{E}_{w_t\sim\mathbb{P}_t} [w_t], \quad \hat{w}_t := \mathbb{E}_{w_t\sim\mathbb{Q}_t} [w_t]
\]
denote the mean vectors of $w_t$ with respect to $\mathbb{P}_t$ and $\mathbb{Q}_t$, respectively.
Similarly, let
\begin{equation}\nonumber
\begin{split}
\Sigma_t &:= \mathbb{E}_{w_t\sim\mathbb{P}_t} [(w_t - \bar{w}_t)(w_t - \bar{w}_t)^\top],\\
 \hat{\Sigma}_t &:= \mathbb{E}_{w_t\sim\mathbb{Q}_t} [(w_t - \hat{w}_t)(w_t - \hat{w}_t)^\top ].
 \end{split}
\end{equation}
denote the covariance matrices of $w_t$ with respect to $\mathbb{P}_t$ and $\mathbb{Q}_t$, respectively.
 Then, the 2-Wasserstein distance between $\mathbb{P}_t$ and $\mathbb{Q}_t$ is bounded by
\[
W_2(\mathbb{P}_t, \mathbb{Q}_t) \geq \sqrt{\|\bar{w}_t - \hat{w}_t\|_2^2 + \mathrm{B}^2(\Sigma_t, \hat{\Sigma}_t)},
\]
where
\[
\mathrm{B}^2(\Sigma_t, \hat{\Sigma}_t) := \mathrm{Tr}[\Sigma_t + \hat{\Sigma}_t - 2(\Sigma^{1/2}_t \hat{\Sigma}_t \Sigma^{1/2}_t)^{1/2}].
\]
Moreover, the bound is exact if $\mathbb{P}_t$ and $\mathbb{Q}_t$ are elliptical distributions with the same density generator.
\end{lemma}

The above lower bound is also known as the Gelbrich bound~\cite[Th.~4]{kuhn2019wasserstein}. 
It evaluates the Wasserstein distance  ignoring higher-order moments, resulting in a tractable form.

The cost function~\eqref{obj} can  be approximated further by
\[
\begin{split}
{J}_\lambda(\pi,\gamma) &=  
\mathbb{E}_{\bold{y}}\Big[\mathbb{E}_{x_T}[x_T^\top Q_f x_T \mid I_T]\\
&+\sum_{t=0}^{T-1} \Big(\mathbb{E}_{x_t, w_t}[x_t^\top Q x_t + u_t^\top R u_t \mid I_t, u_t] \\
& \qquad \qquad- \lambda\big[\|\bar{w}_t - \hat{w}_t\|^2 + \mathrm{B}^2(\Sigma_t,\hat{\Sigma}_t)\big]\Big)\Big].
\end{split}
\]
Then, the following minimax  control problem approximates the original WDRC problem~\eqref{drc}:
\begin{equation}\label{minimax}
\min_{\pi\in\Pi}\max_{\gamma\in\Gamma} J_\lambda(\pi, \gamma),
\end{equation}
where the new policy space for the adversary is defined as
$\Gamma := \{\gamma := (\gamma_0, \ldots, \gamma_{T-1}) \mid \gamma_t (I_t) = \mathbb{P}_t \in \mathcal{P} (\mathbb{R}^{n_x}) \}$ and no longer depends on the ambiguity sets.
Let $(\pi^*,\gamma^*)$ denote the optimal policy pair of the approximate problem.
The performance gap between $\pi^*$ and the optimal DR policy $\pi^{opt}$  will be discussed in Section~\ref{sec:perf}.
Before that, a closed-form expression of $(\pi^*,\gamma^*)$ will be derived using mathematical induction in the following subsection.  


\subsection{Solution via Riccati Equation}

Let the value function for the approximate problem be recursively defined by $V_T(I_T) = \mathbb{E}_{x_T}[x_T^\top Q_f x_T \mid I_T]$ and
\begin{equation}\label{vf}
\begin{split}
& V_t (I_t) = \min_{u_t}\max_{\bar{w}_t, \Sigma_t\succeq 0}\mathbb{E}_{\substack{x_t, y_{t+1}, \\  w_t\sim \mathbb{P}_t}}\bigg[x_t^\top Q x_t + u_t^\top R u_t\\
& - \lambda [\|\bar{w}_t - \hat{w}_t\|^2 + \mathrm{B}^2(\Sigma_t,\hat{\Sigma}_t)]  + V_{t+1}(I_t, y_{t+1}, u_t) \mid I_t, u_t\bigg]
\end{split}
\end{equation}
for $t=T-1, \dots,0$.
Then, according to the DP principle (e.g., 
{~\cite{osogami2015robust, saghafian2018ambiguous}}), we have
\begin{equation}\label{opt_cost}
\inf_{\pi\in\Pi}\sup_{\gamma\in\Gamma} J_\lambda(\pi, \gamma) = \mathbb{E}_{y_0}[V_0(I_0)] =: J^*_\lambda.
\end{equation}
 
Let 
\[
\bar{x}_{t} := \mathbb{E}_{x_t}[x_{t} \mid I_{t}], \; \; \bar{P}_{t} :=\mathbb{E}_{x_t}[(x_{t}-\bar{x}_t)(x_{t}-\bar{x}_t)^\top \mid I_{t}]
\]
 denote the expected value and the covariance matrix of the state, respectively, conditioned on the information available at time $t$ and let
 $\xi_{t} := x_{t} - \bar{x}_{t}$
  be the difference between the actual state and its expected value. Also, let 
\[
\Phi: = B R^{-1} B^\top - \frac{1}{\lambda} I \in \mathbb{S}^{n_x}.
\]
Under the condition that $V_{t+1}$ has a quadratic form, 
we obtain an explicit solution of the minimax optimization problem in the Bellman recursion~\eqref{vf}.

\begin{lemma}\label{lem:opt_cont}
Suppose that
\begin{equation}\label{vf_next}
\begin{split}
V_{t+1}(I_{t+1}) =\, &\mathbb{E}_{x_{t+1}}[ x_{t+1}^\top P_{t+1}x_{t+1} + \xi_{t+1}^\top S_{t+1}\xi_{t+1}\\
&+ 2r_{t+1}^\top x_{t+1} \mid I_{t+1}] + z_{t+1} + \sum_{s=t+1}^{T-1} \tilde{z}_{s},
\end{split}
\end{equation}
for some $P_{t+1}\in\mathbb{S}_{+}^{n_x}, S_{t+1}\in{\mathbb{S}^{n_x}_+}, r_{t+1}\in\mathbb{R}^{n_x}, z_{t+1}\in\mathbb{R}$, and $\tilde{z}_{s}\in\mathbb{R}, s=t+1,\dots, T-1$.
Assume further that the penalty parameter satisfies the condition $\lambda I \succ P_{t+1}$. Then,  the outer minimization problem has the following unique solution:
\begin{equation}\label{u_opt}
u_t^* = K_t \bar{x}_t + L_t,
\end{equation}
where
\begin{equation}\label{contr_params}
\begin{split}
K_t &= - R^{-1}B^\top(I+P_{t+1}\Phi)^{-1} P_{t+1} A,\\
L_t &= - R^{-1}B^\top(I+P_{t+1}\Phi)^{-1}  (P_{t+1} \hat{w}_t + r_{t+1}).
\end{split}
\end{equation}
Moreover, given $u_t^*$,
the inner maximization problem with respect to $\bar{w}_t$ has the following unique solution:
\begin{equation}\label{mu_st}
    \bar{w}_t^*  = (\lambda I - P_{t+1})^{-1} \big(r_{t+1} + P_{t+1}[A \bar{x}_t + B u_t^*] + \lambda\hat{w}_t\big).
\end{equation}
\end{lemma}
The proof of this lemma can be found in Appendix~\ref{app:opt_cont}.
It follows from the above lemma that if $V_{t+1}$ is in the quadratic form, $V_t$ will also be quadratic under a certain condition on $\lambda$ similar to~\cite[Assumption 1]{kim2020minimax}. More specifically, we impose the following assumption on the penalty parameter.

\begin{assumption}\label{ass:lambda_ass}
The penalty parameter satisfies  $\lambda I \succ P_{t}$ for all $t=1,\dots, T$.
\end{assumption}

Under this assumption, we can use mathematical induction backward in time to recursively  show that the value functions $V_t$  are quadratic for all $t$.
Accordingly, a closed-form expression of the optimal policy is obtained. The following theorem formalizes the results.

\begin{theorem}\label{thm:sol}
Suppose that Assumption~\ref{ass:lambda_ass} holds. Then, the value function at time $t = 0, \ldots, T$ has the following quadratic form:
\[
V_t(I_t) = \mathbb{E}_{x_t}[x_t^\top P_t x_t + \xi_t^\top S_{t}\xi_t + 2 r_t^\top x_t \mid I_t] + z_{t} + \sum_{s=t}^{T-1}\tilde{z}_s.
\]
Here, the coefficients $P_{t+1}\in\mathbb{S}_{+}^{n_x}, S_{t+1}\in\mathbb{S}^{n_x}_+, r_{t+1}\in\mathbb{R}^{n_x}$ and $z_{t+1}\in\mathbb{R}$ can be
found recursively  using a Riccati equation, defined as
\begin{align}
    P_t  =\, & Q + A^\top (I + P_{t+1}\Phi)^{-1} P_{t+1} A \label{P}\\
    S_t  = \, & Q + A^\top  P_{t+1} A - P_t\label{S}\\
    r_t  =  \, & A^\top (I + P_{t+1}\Phi)^{-1} ( r_{t+1} + P_{t+1}\hat{w}_t)\label{r}\\
    z_t = \, & z_{t+1}  + (2\hat{w}_t - \Phi r_{t+1})^\top(I + P_{t+1}\Phi)^{-1} r_{t+1}\nonumber\\
    &+\hat{w}_t^\top (I + P_{t+1}\Phi)^{-1}  P_{t+1} \hat{w}_t - \lambda\mathrm{Tr}[\hat{\Sigma}_t]\label{z} 
\end{align}
with the terminal conditions $P_T = Q_f, S_T = \mathbf{0}_{n_x\times n_x}, r_T = \mathbf{0}_{n_x}$ and $z_T = 0$.
The coefficients $\tilde{z}_t, t=0,\dots, T-1$, can be found recursively as the optimal value of the following maximization problem:
\begin{equation}\label{z_tilde}
\begin{split}
\tilde{z}_t:=\max_{\Sigma_t\succeq 0}\mathrm{Tr}[S_{t+1} \bar{P}_{t+1}] &+ \mathrm{Tr}[(P_{t+1} - \lambda I) \Sigma_t] \\
&+ 2 \lambda \mathrm{Tr}[(\Sigma_t^{1/2}\hat{\Sigma}_t\Sigma_t^{1/2})^{1/2}].
\end{split}
\end{equation}

Moreover, the optimal control policy is unique and is obtained as
\begin{equation}\label{opt_pol}
\pi^*_t(I_t) = K_t\bar{x}_t + L_t,
\end{equation}
with $K_t$ and $L_t$ defined in~\eqref{contr_params}.
\end{theorem}

The proof of this theorem can be found in Appendix~\ref{app:sol}.
As a corollary, 
any worst-case disturbance distribution generated by 
 the optimal distribution policy $\gamma^*(\bar{x}_t)$ has the mean given by \eqref{mu_st}, while its covariance is the optimal solution of~\eqref{z_tilde}.

It is worth noting that \eqref{z_tilde} requires the computation of the conditional covariance matrix $\bar{P}_{t+1}$ of $x_{t+1}$. Therefore, there is a need for an estimator to produce a state estimate given observations.
As one of the most common state estimator for linear systems, we use the Kalman filter~\cite{kalman1960new}. Specifically, in each time stage $t+1$, given a control input $u_t$, an observation $y_{t+1}$ and a disturbance distribution with mean $\bar{w}_t$ and covariance $\Sigma_t$, the expected value of state $x_{t+1}$ is computed as
\[
\bar{x}_{t+1} = \bar{x}_{t+1|t} + \bar{P}_{t+1}C^\top M^{-1} (y_{t+1} - C \bar{x}_{t+1|t}),
\]
where $\bar{x}_{t+1|t} = A\bar{x}_{t} + B u_t + \bar{w}_t$. The conditional covariance matrix $\bar{P}_{t+1}$ is computed  recursively as follows:
\begin{equation}\label{post_cov}
    \bar{P}_{t+1}  =  \bar{P}_{t+1|t} - \bar{P}_{t+1|t} C^\top(C \bar{P}_{t+1|t} C^\top + M)^{-1}C \bar{P}_{t+1|t},
    \end{equation}
    where
    \begin{equation}\label{prior_cov}
    \bar{P}_{t+1|t} = A \bar{P}_t A^\top + \Sigma_t.
\end{equation}
Note that the state estimates in our case are computed using the mean vector $\bar{w}_t^*$ and the covariance matrix $\Sigma_t^*$ of the worst-case disturbance distribution. 

When using the Kalman filter,  the intractability of~\eqref{z_tilde} can be tackled by
 leveraging the conditional covariance equations~\eqref{post_cov} and~\eqref{prior_cov}. Specifically, the problem~\eqref{z_tilde} can be reformulated as the following tractable SDP problem:
\begin{equation}\label{sdp}
\begin{split}
\max_{\substack{U,V\succeq 0,\\\Sigma_t\succeq 0}} \; & \mathrm{Tr}[S_{t+1} V] + \mathrm{Tr}[(P_{t+1}-\lambda I)\Sigma_t] + 2\lambda \mathrm{Tr}[U]\\
\mbox{s.t.} \; &  \begin{bmatrix} \hat{\Sigma}_t^{1/2} \Sigma_t \hat{\Sigma}_t^{1/2} & U \\ U & I\end{bmatrix} \succeq 0\\
&\begin{bmatrix}\bar{P}_{t+1\mid t} - V & \bar{P}_{t+1\mid t} C^\top \\ C \bar{P}_{t+1 \mid t} & C \bar{P}_{t+1\mid t} C^\top + M\end{bmatrix}\succeq 0\\
&C \bar{P}_{t+1\mid t} C^\top + M \succeq 0\\
& \bar{P}_{t+1\mid t} = A \bar{P}_t A^\top + \Sigma_t,
\end{split}
\end{equation}
where we used the property that $\mathrm{Tr}[S_{t+1}V] \leq \mathrm{Tr}[S_{t+1}V']$ for any $V\preceq V'$ and applied the Schur complement lemma to replace the inequality constraints with corresponding linear matrix inequality ones. The reformulated problem~\eqref{sdp} can be efficiently solved using off-the-shelf solvers~\cite{o2016conic, andersen2003implementing, aps2019mosek}. As a result, the constant $\tilde{z}_t$ and the worst-case covariance matrix $\Sigma_t^*$ can be found forward in time, starting from $\bar{P}_0$.

These results lead us to a WDRC algorithm for partially observable linear systems, described as follows.
In the backward pass, 
we first iteratively solve the Riccati equation~\eqref{P}--\eqref{z}  for $t=T-1, \dots, 0$ to construct the control policy $\pi_t^*$ using Theorem~\ref{thm:sol}. 
In the forward pass, the estimates $\bar{x}_{0}$ and $\bar{P}_{0}$ are initially obtained using the initial state distribution $f_x(x_0)$ and the observation $y_0$. Then, at each time $t$, the  mean $\bar{w}_t^*$ of the worst-case disturbance distribution is obtained as~\eqref{mu_st}, while its covariance matrix $\Sigma_t^*$ is computed by solving the SDP problem~\eqref{sdp}. 
The policy $\pi_t^*(I_t)$ in~\eqref{opt_pol} is then applied to the system. Finally, the estimates $\bar{x}_{t+1}$ and $\bar{P}_{t+1}$ are updated using the worst-case disturbance distribution and the new observation $y_{t+1}$.

\section{Performance Guarantee}\label{sec:perf}

In this section, we compare the performances of the approximate policy $\pi^*$ and the optimal DR policy $\pi^{opt}$.
We first show that the total cost incurred by the approximate policy is bounded above by the optimal cost \eqref{opt_cost} of the approximate problem
 plus a constant. 
 This guaranteed cost property is then used to obtain a provable bound for the optimality gap.

\subsection{Guaranteed Cost Property}

The following proposition indicates the guaranteed cost property of our approximate policy $\pi^*$ for any worst-case disturbance distributions selected from the Wasserstein ambiguity sets.

\begin{proposition}\label{prop:gc}
Given $\lambda > 0$, 
let $\pi_\lambda^*$ be the optimal control policy of the approximate problem~\eqref{minimax}.
Then, the cost incurred by $\pi_\lambda^*$ under the worst-case distribution policy in $\Gamma_{\mathcal{D}}$ is bounded  as follows:
\[
 \sup_{\gamma \in \Gamma_{\mathcal{D}}} J(\pi_\lambda^*, \gamma) \leq \lambda T \theta^2 +  J^*_\lambda,
\]
where  $J^*_\lambda$ is the optimal cost~\eqref{opt_cost} of the approximate problem. 
\end{proposition}

\begin{proof}
Although the proof is similar to that of \cite[Lemma~3, Th.~6]{Kim2021} for the full observation case, we have included it for the completeness of the paper. 
Fix $\lambda > 0$. 
Let $\mathrm{LHS} := \sup_{\gamma \in \Gamma_{\mathcal{D}}} J(\pi_\lambda^*, \gamma)$ and 
$\mathrm{RHS} := \lambda T \theta^2 +  J^*_\lambda$. 
For any $\epsilon > 0$, there exists $\gamma^\epsilon \in \Gamma_{\mathcal{D}}$ such that
\[
\mathrm{LHS} - \epsilon < J(\pi_\lambda^*, \gamma^\epsilon).
\]
Recall that 
\[
\| \bar{w}_t - \hat{w}_t \|_2^2 + \mathrm{B}^2(\Sigma_t, \hat{\Sigma}_t ) \leq W_2(\mathbb{P}_t, \mathbb{Q}_t)^2 \leq \theta^2.
\]
Thus, it follows from  $\gamma^\epsilon \in \Gamma_{\mathcal{D}}$ that
\begin{equation}\nonumber
\begin{split}
J (\pi_\lambda^*, \gamma^\epsilon) &\leq \lambda T \theta^2 + J_\lambda (\pi_\lambda^*, \gamma^\epsilon)\\
& \leq \lambda T \theta^2 + \sup_{\gamma \in \Gamma} J_\lambda (\pi_\lambda^*, \gamma)=  \lambda T \theta^2 + J^*_\lambda. 
\end{split}
\end{equation}
Since $\epsilon$ was chosen arbitrarily, 
$\mathrm{LHS} \leq \mathrm{RHS}$.
\end{proof}

This proposition implies the distributional robustness of the optimal control policy $\pi_\lambda^*$ of our approximate problem. 
In the following subsection, we use this property to bound the performance gap between $\pi_\lambda^*$ and the optimal DR policy.

\subsection{Optimality Gap}

Let $V^\mathrm{DR}_t$ denote the optimal value function of the original WDRC problem~\eqref{drc}, recursively defined similar to~\eqref{vf}. 
Then,  the DP principle yields
\[
J^*_\mathrm{DR} := \inf_{\pi \in \Pi} \sup_{\gamma \in \Gamma_{\mathcal{D}}} J(\pi, \gamma) =\mathbb{E}_{y_0}[V_0^\mathrm{DR}(I_0)],
\]
which corresponds to the optimal cost incurred by $\pi^{opt}$.
Thus,  $V_0^\mathrm{DR} (I_0)$ represents the performance of the optimal policy in the WDRC setting. 

Similarly, let $V^\mathrm{LQ}_t$ denote the optimal value function of the standard LQ control problem with the nominal distribution.
Again by the DP principle, we have
\begin{equation}\nonumber
\begin{split}
&J^*_\mathrm{LQ} = \inf_{\pi \in \Pi}\mathbb{E}_{\bold{y}}\bigg [ \mathbb{E}_{x_T}[x_T^\top Q_f x_T \mid I_T] \:  \\ 
&
+ \sum_{t=0}^{T-1} \mathbb{E}_{x_t, w_t \sim \mathbb{Q}_t}[  x_t^\top Q x_t + u_t^\top R u_t \mid I_t, u_t] \bigg ].
\end{split}
\end{equation}
The LQ control problem can be interpreted as the particular case of the WDRC problem in which the adversary is devoted to the nominal distribution policy $\hat{\gamma} (I_t) = \mathbb{Q}_t \in \mathcal{D}_t$.
Since the nominal distribution policy is admissible, that is, $\hat{\gamma} \in \Gamma_{\mathcal{D}}$, we have
$J^*_\mathrm{LQ} \leq J^*_\mathrm{DR}$.

Using the relationships between the standard LQ control, the WDR control and our approximation, we obtain the following bound for the optimality gap.

\begin{theorem}
Given $\lambda > 0$, the optimal control policy $\pi_\lambda^*$ of the approximate problem~\eqref{minimax} has the following relative performance guarantee:
\[
J^*_{\mathrm{DR}} \leq \sup_{\gamma \in \Gamma_{\mathcal{D}}} J(\pi_\lambda^*, \gamma) \leq \rho J^*_{\mathrm{DR}},
\]
where the relative performance guarantee $\rho$ is given by
\begin{equation}\label{perf_guar}
\rho:= \frac{\lambda T \theta^2 + {J^*_\lambda}}{J^*_\mathrm{LQ}} > 1.
\end{equation}
\end{theorem}

\begin{proof}
By the definition of the optimal value function, ${J^*_\mathrm{DR}}\leq \sup_{\gamma \in \Gamma_{\mathcal{D}}} J(\pi_\lambda^*, \gamma)$.
Let the nominal distribution policy $\hat{\gamma}_t$ be defined by
\[
\hat{\gamma}_t (I_t) = \mathbb{Q}_t \in \mathcal{D}. 
\]
Then, $\hat{\gamma} \in \Gamma_{\mathcal{D}}$. 
Thus, we have
\[
{J^*_\mathrm{LQ}} = \inf_{\pi \in \Pi} J(\pi, \hat{\gamma}) \leq {J^*_\mathrm{DR}}.
\]
It follows from Proposition~\ref{prop:gc} that
\begin{equation}\nonumber
\sup_{\gamma \in \Gamma_{\mathcal{D}}} J(\pi_\lambda^*, \gamma) \leq  \lambda T \theta^2 + {J^*_\lambda}= \rho {J^*_\mathrm{LQ}} \leq \rho {J^*_\mathrm{DR}}. 
\end{equation}
\end{proof}

The relative performance guarantee $\rho$ can be computed by solving the standard LQ control problem with the nominal distribution and our approximate problem~\eqref{minimax}.  
Both problems are tractable to solve via Riccati equations, unlike the original WDRC problem. 
It is often desirable to select $\lambda$ such that it minimizes $\rho$ or, equivalently, $\lambda T \theta^2 + {J^*_\lambda}$.
Since $\lambda$ is a scalar, such a  $\lambda$ can be efficiently obtained by a binary search.

\section{Simulation Results}\label{sec:exp}

In this section, we demonstrate the performance of the proposed method using $(i)$ a Gaussian disturbance and $(ii)$ a non-Gaussian disturbance. We compare our WDRC method with the standard partially observable LQG control, which uses estimated distributions of the disturbances.

More specifically, we consider a discrete-time system with the following parameters:
\[
    A = \begin{bmatrix}
    0.518 & 0.266\\
    0.405 & 0.806
    \end{bmatrix}, 
    B = \begin{bmatrix}
    -2.972 \\ -2.271
    \end{bmatrix},
C = \begin{bmatrix}1.023 & 1.955\end{bmatrix},
\]
which is unstable due to an eigenvalue outside the unit circle. The controller is required to minimize the cost with parameters $Q = Q_f = R = I$ over the time horizon of $T=50$.

The nominal disturbance distribution of $w_t$ is estimated as a Gaussian with the following mean and covariance matrix constructed  from $N=5$ sample data:
\[
\hat{w}_t = \frac{1}{N}\sum_{i=1}^{N} \tilde{w}_t^{(i)}, \; \hat{\Sigma}_t = \frac{1}{N}\sum_{i=1}^{N} (\tilde{w}_t^{(i)} - \hat{w}_t)(\tilde{w}_t^{(i)} - \hat{w}_t)^\top,
\]
where $\tilde{w}_t^{(i)}$ is the $i$th sample disturbance drawn from the true probability distribution. The states are estimated via the Kalman filter.\footnote{Since the actual disturbance distribution is unknown, the mean and covariance of the nominal distribution are used in the Kalman filter for the standard LQG.}

All algorithms were implemented in Python and run on a PC with an Intel Core i7-8700K (3.70 GHz) CPU and 32 GB RAM.\footnote{The source code of our implementation is available online: \href{https://github.com/CORE-SNU/PO-WDRC}{\tt https://github.com/CORE-SNU/PO-WDRC}}

\subsection{Gaussian Disturbances}

\begin{figure}[t]
    \centering
    \centering
    \includegraphics[height=2in]{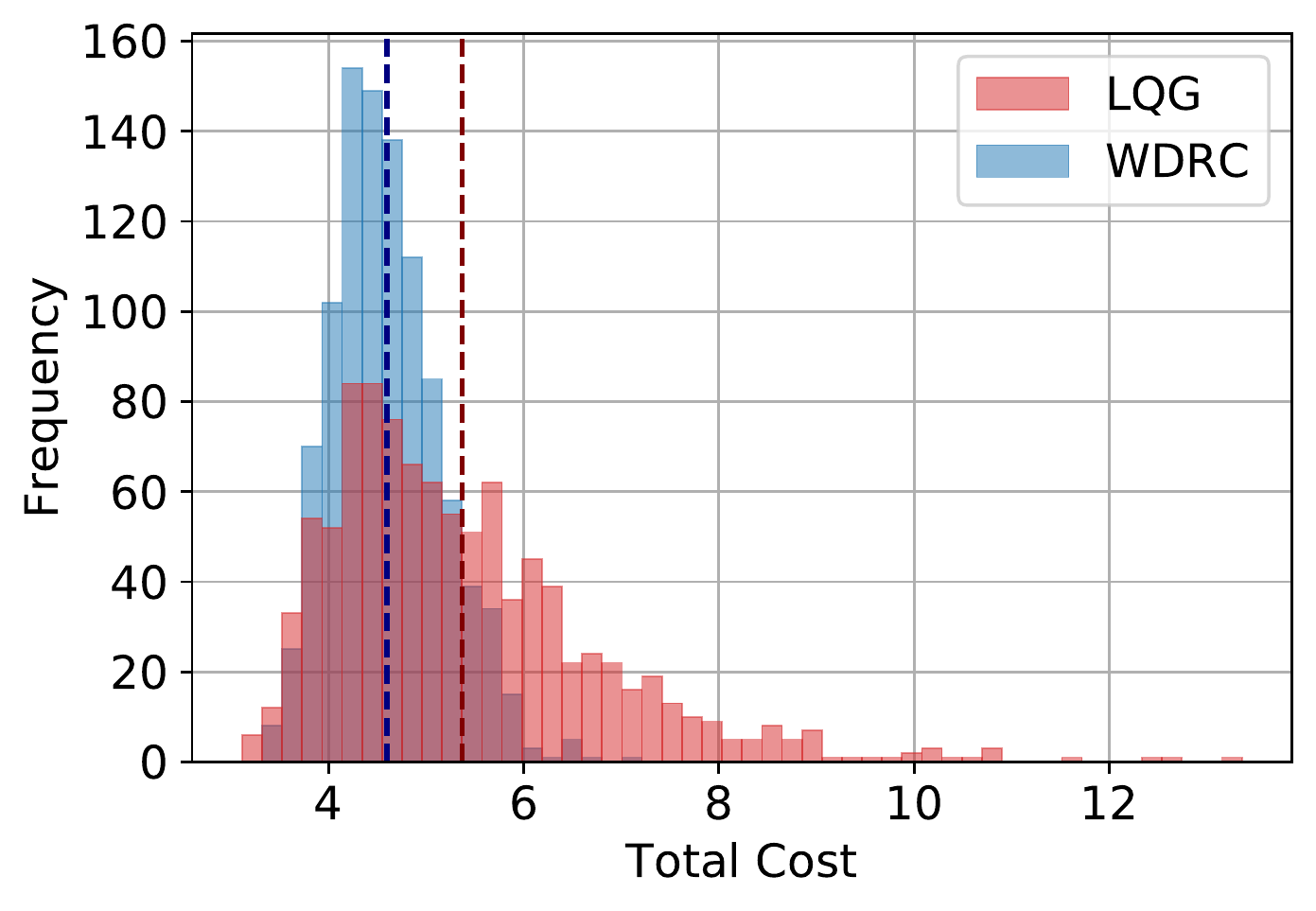}
    \caption{Histogram of the total costs in the case of Gaussian disturbances. The dashed lines represent the sample means of the costs returned by the two methods.}
    \label{fig:hist_g}
\end{figure}

\begin{table}[t]
\caption{Total costs for the Gaussian and the Uniform disturbance cases.}
\centering
\setlength{\tabcolsep}{0.2em} 
\begin{tabular}{>{\raggedright}c| >{\centering}m{1.6cm} | >{\centering}m{1.4cm} | >{\centering}m{1.4cm} | >{\centering\arraybackslash}m{1.4cm}}
\hline
\multirow{2}{*}{\textbf{Algorithm}} & \multicolumn{2}{c|}{Gaussian} & \multicolumn{2}{c}{Uniform}\\
\cline{2-5}
& Mean & Std. Dev. & Mean & Std. Dev.\\
\hline
\textbf{WDRC} & 4.599 & 0.557 & 0.536 & 0.151\\
\cline{1-5}
\textbf{LQG} & 5.374 & 1.398 & 0.781 & 0.267\\\hline
\end{tabular}
\label{table:costs}
\end{table}

In the first scenario, the true disturbance distribution is chosen as $\mathcal{N}([0.01, 0.02]^\top, [0.01,\, 0.005; 0.005,\, 0.01])$, and the disturbance data $\tilde{w}_t^{(i)}$ are sampled from this distribution.
The observation noise $v_t$ follows 
a zero-mean Gaussian distribution with covariance $M=0.2 I$, and
the initial state is assumed to be distributed according to $f_x(x_0) = \mathcal{N}([-1, -1]^\top, 0.001 I)$.
 The penalty parameter was calibrated by minimizing the upper-bound in Proposition~\ref{prop:gc} for $\theta = 0.1$ so that it satisfies Assumption~\ref{ass:lambda_ass}.


Fig.~\ref{fig:hist_g} shows  the distributions of the total costs over 1,000 simulations as a histogram.
Overall, the cost distribution for the WDR controller has a bell shape, and thus
  is more favorable than that for the LQG controller.
The WDR controller returns lower costs with a higher probability compared to the LQG controller. 
This is explained by the fact that the WDR controller anticipates mismatches between the true disturbance distribution and the nominal one. Meanwhile, LQG is unable to deal with such unexpected distribution errors, causing higher total costs with a right-skewed distribution. In addition, the WDR controller is less sensitive to the state estimates $\bar{x}_t$, unlike LQG, which relies solely on the inaccurate nominal distribution at both the control and estimation stages. 

The total costs for both WDR and LQG methods are reported in Table~\ref{table:costs}. 
The WDR controller incurs a lower average total cost with a smaller standard deviation compared to the LQG method, confirming the superiority of our method.

\subsection{Non-Gaussian Disturbances}

\begin{figure}[t]
    \centering
    \centering
    \includegraphics[height=2in]{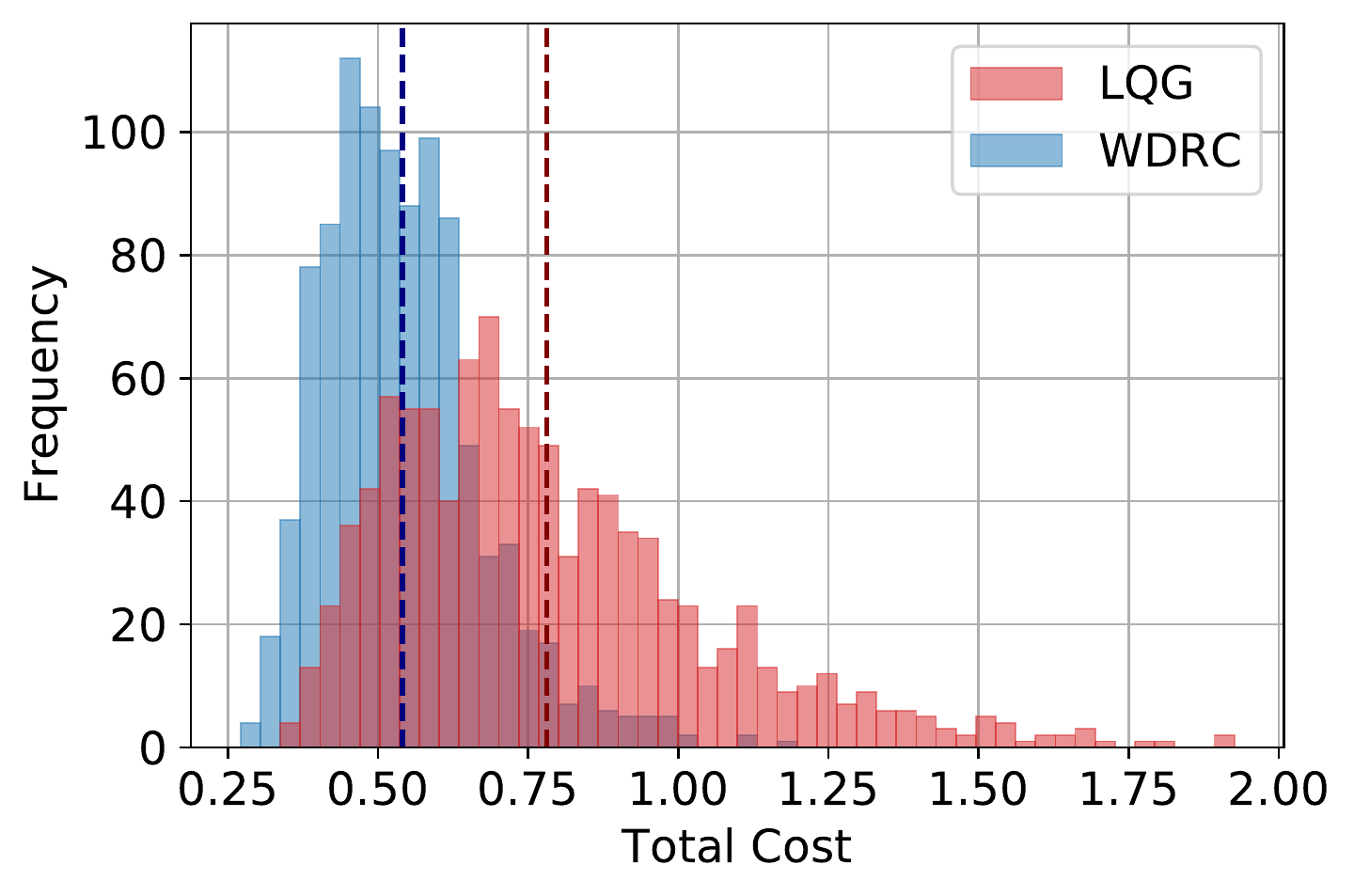}
    \caption{Histogram of the total costs in the case of uniform disturbances. The dashed lines represent the sample means of the costs returned by the two methods.}
    \label{fig:hist_u}
\end{figure}

In the second scenario, the true disturbance distribution is assumed to be uniform, $\mathcal{U}[-0.05, 0.05]^{2}$, and the disturbance data $\tilde{w}_t^{(i)}$ are sampled from this distribution.
The observation noise $v_t$ is drawn from a zero-mean Gaussian distribution with covariance $M =  0.1I$.
The initial state is uniformly distributed with $f_x(x_0) = \mathcal{U}([0.1, 0.2]^\top,[0.3, 0.5]^\top)$.
Since the state distribution is not Gaussian in this setting, the Kalman filter is no longer an optimal estimator. Yet, we apply the Kalman filter with
the Gaussian nominal distribution of disturbances
to demonstrate the capability of the WDR controller to compensate for an inexact state estimator. 
The penalty parameter was tuned for $\theta=0.03$ following  the same procedure as in the Gaussian case.


Fig.~\ref{fig:hist_u} illustrates the cost histograms over 1,000 simulation runs. The effect of the penalty term is more pronounced here, as the difference between the cost distributions is larger compared to the Gaussian case. 
 In particular, the total costs incurred by the WDR controller are concentrated in the low-cost regions, while those incurred by LQG are spread wider, with a right tail in the high-cost region.

Table~\ref{table:costs} summarizes the total costs for both WDR and LQG methods. Analogous to the Gaussian case, the average total cost incurred by the WDR controller is significantly lower than that obtained using LQG. Moreover, the standard deviation of the costs is considerably smaller when using the WDR controller.
 The reason for this result is twofold. First, the nominal distribution is not an efficient estimator of the true uniform distribution;   therefore, relying on moment estimates is insufficient. The WDR controller alleviates this issue by considering all distributions close to the nominal one, thereby enabling the system to effectively handle the distribution mismatch. 
 Second, the state estimation for LQG is performed for Gaussian disturbances with a nominal mean and covariance, while the WDRC method uses the worst-case distribution in the state estimation, adding  additional robustness to the estimation stage.
%

%

\section{Conclusions and Future Work}

In this work, we have presented a novel WDRC method for 
 discrete-time partially observable linear systems. 
 We proposed a  tractable reformulation of the original WDRC problem by recasting it as a  minimax control problem with an approximate Wasserstein penalty. 
 Furthermore, we derived a closed-form expression of the optimal control policy for the approximate problem with corresponding Riccati equations. Regarding the worst-case disturbance distribution, we obtain the closed-form solution for its mean and formulate a tractable SDP problem for its covariance matrix.
 The proposed method has several salient features, such as robustness to inaccurate distribution information about the disturbances and a provable bound for the optimality gap. 
 The experiment results demonstrate the capability of our method to hedge against distributional uncertainties. 

In the future, we aim to extend our results to the infinite-horizon setting and examine closed-loop stability. 
Moreover, the proposed WDRC method can be used in conjunction with robust state estimators to explicitly account for distributional errors in the sensor noise as well. 

\section*{Acknowledgement}

We thank the anonymous reviewers for their insightful comments.

\appendices
\section{Proofs}
\subsection{Proof of Lemma~\ref{lem:opt_cont}}\label{app:opt_cont}
\begin{proof}
It follows from~\eqref{vf} and~\eqref{vf_next} that
\[
\begin{split}
&V_t(I_t)  = \min_{u_t} \max_{\bar{w}_t, \Sigma_t\succeq 0}  \mathbb{E}_{x_t}[x_t^\top Q x_t \mid I_t] + \mathbb{E}_{x_{t+1}} [x_{t+1}^\top P_{t+1} x_{t+1} \\
& + 2 r_{t+1}^\top x_{t+1}\mid I_{t}, u_{t}] + \mathbb{E}_{x_{t+1}, y_{t+1}}[\xi_{t+1}^\top S_{t+1} \xi_{t+1} \mid I_{t}, u_{t}] \\
&+u_t^\top R u_t+ z_{t+1}+ \sum_{s=t+1}^{T-1}\tilde{z}_{s}- \lambda[\|\bar{w}_t - \hat{w}_t\|^2 + \mathrm{B}^2(\Sigma_t, \hat{\Sigma}_t)].
\end{split}
\]
We note that
\[
\mathbb{E}_{x_{t+1}, y_{t+1}}[\xi_{t+1}^\top S_{t+1} \xi_{t+1} \mid I_{t}, u_{t}] = \mathrm{Tr}[S_{t+1} \bar{P}_{t+1}],
\]
which is independent of $u_t$ and $\bar{w}_t$.
Therefore,
\[
\begin{split}
&V_t( I_t)  =\min_{u_t} \max_{\bar{w}_t, \Sigma_t\succeq 0} \mathbb{E}_{x_t}[x_t^\top Q x_t \mid I_t] + u_t^\top R u_t\\
&+ \mathbb{E}_{x_t}[(Ax_t + B u_t + \bar{w}_t)^\top P_{t+1}(Ax_t + B u_t + \bar{w}_t) \mid I_{t}, u_{t}] \\
&+ 2 r_{t+1}^\top (A\bar{x}_t + B u_t + \bar{w}_t) + z_{t+1} + \sum_{s=t+1}^{T-1}\tilde{z}_{s}\\
& - \lambda[\|\bar{w}_t - \hat{w}_t\|^2 + \mathrm{B}^2(\Sigma_t, \hat{\Sigma})]\\
& + \mathrm{Tr}[P_{t+1} \Sigma_t] + \mathrm{Tr}[S_{t+1}\bar{P}_{t+1}],
\end{split}
\]
where we used $\mathbb{E}[w_t^\top P_{t+1} w_t] = \bar{w}_t^\top P_{t+1} \bar{w}_t + \mathrm{Tr}[P_{t+1}\Sigma_t]$.

Note that the maximization problem with respect to $\bar{w}_t$ is now separated from the one with respect to $\Sigma_t$, enabling to solve each problem independently.  Differentiating the objective function with respect to $\bar{w}_t$, the first-order optimality condition yields an optimal solution $\bar{w}_t^*$ satisfying
\[
(\lambda I - P_{t+1})\bar{w}_t^* = P_{t+1}[A_t \bar{x}_t + B_t u_t] + r_{t+1} + \lambda\hat{w}_t.
\]
It follows from the assumption on penalty parameter that  the Hessian of value function with respect to $\bar{w}_t$ is negative definite. Therefore, the objective is strictly concave and has a unique maximizer~\eqref{mu_st}.


On the other hand, by grouping terms depending on the covariance $\Sigma_t$, it is observed that the worst-case covariance $\Sigma_t^* \succeq 0$  maximizes
\begin{equation}\label{cov_obj}
\mathrm{Tr}[S_{t+1} \bar{P}_{t+1}] + \mathrm{Tr}[(P_{t+1} - \lambda I) \Sigma_t] + 2 \lambda \mathrm{Tr}[(\Sigma_t^{1/2}\hat{\Sigma}_t\Sigma_t^{1/2})^{1/2}].
\end{equation}
Note that $\Sigma_t^*$ and \eqref{cov_obj} are independent of the control input.
By differentiating the objective function of the outer minimization problem
with respect to $u_t$,  we obtain
\begin{equation}\label{opt_cond}
\begin{split}
& 2\left[B + \frac{\partial \bar{w}^*_t}{\partial u_t}\right]^\top [P_{t+1}(A\bar{x}_t + B u_t + \bar{w}_t^*) + r_{t+1}] \\
&- 2\lambda \frac{\partial\bar{w}_t^*}{\partial u_t}(\bar{w}_t^* - \hat{w}_t) + 2R u_t = 2 B^\top g_t(u_t) + 2 R u_t
\end{split}
\end{equation}
where
\begin{equation}\label{g_u}
g_t(u_t) := P_{t+1}(A\bar{x}_t + B u_t + \bar{w}_t^*) + r_{t+1}.
\end{equation}
Under the assumption on the penalty parameter $\lambda$, the Hessian of the objective function with respect to $u_t$ is positive definite.
Thus, the outer minimization problem has the following unique solution:
\begin{equation}\label{cont}
    u_t^* = - R^{-1} B^\top g_t^*,
\end{equation}
where $g^* = g^*(u^*)$. In order to solve~\eqref{cont}, we rewrite the  mean of the worst-case disturbance distribution as
\[
\bar{w}_t^* = \frac{1}{\lambda}\big(P_{t+1}(A\bar{x}_t + Bu_t^* +\bar{w}_t^*) + r_{t+1} + \lambda\hat{w}_t\big){\color{red}.}
\]
Then, substituting it into~\eqref{g_u}, we obtain
\[
g_t^* = (I + P_{t+1}\Phi)^{-1} (P_{t+1} A\bar{x}_t + P_{t+1} \hat{w}_t + r_{t+1}).
\]
We conclude the proof by replacing the above expression into~\eqref{cont}.
\end{proof}

\subsection{Proof of Theorem~\ref{thm:sol}}\label{app:sol}
\begin{proof}

The theorem can be proved via mathematical induction. At $t=T$, the statement holds as $V_T(I_T) = \mathbb{E}_{x_T}[x_T^\top P_T x_T | I_T]$. Now suppose that the induction hypothesis is true for $t+1$. Then, Lemma~\ref{lem:opt_cont} provides that at time $t$, the inner maximization with respect to the disturbance mean attains an optimal solution $\bar{w}_{t}^*$ computed by~\eqref{mu_st}, while the outer problem is minimized by $u_t^*$ calculated according to~\eqref{u_opt}. Plugging these values into $V_t(I_t)$, we have
\[
\begin{split}
V_t(I_t)  = \, &\mathbb{E}_{x_t}[x_t^\top (Q + A^\top P_{t+1} A) x_t \mid I_t]\\
& - \bar{x}_{t}^\top S_t \bar{x}_t + 2r_t^\top \bar{x}_t + z_{t} + \tilde{z}_t + \sum_{s=t+1}^{T-1}\tilde{z}_{s},
\end{split}
\]
where $S_t = A^\top P_{t+1} \Phi(I+P_{t+1}\Phi)^{-1} P_{t+1} A$, while $r_t$ and $z_t$ are defined in~\eqref{r} and~\eqref{z}, and $\tilde{z}_{t}$ is the optimal value of~\eqref{z_tilde}.
Simplifying the expressions, we have
\[
\begin{split}
V_t(I_t)  = \, &  \mathbb{E}_{x_t}[x_t^\top (Q + A^\top P_{t+1} A - S_{t}) x_t\mid I_t]\\
&+ \mathbb{E}_{x_t}  [\xi_t^\top S_t\xi_t + 2r_t^\top x_t\mid I_t]  + z_{t} + \sum_{s=t}^{T-1}\tilde{z}_{s}\\
= \, &\mathbb{E}_{x_t}[x_t^\top P_t x_t + \xi_t^\top S_{t}\xi_t + 2 r_t^\top x_t \mid I_t] + z_{t} + \sum_{s=t}^{T-1}\tilde{z}_{s},
\end{split}
\]
which is in the required quadratic form with parameters~\eqref{P}--\eqref{z}. This completes our inductive argument, and the result follows. 
\end{proof}

\bibliographystyle{IEEEtran}
\bibliography{reference} 

\end{document}